\documentclass[12pt,letterpaper]{article}
\usepackage{amsmath,amsfonts,amsthm,amssymb,tabls,stmaryrd,graphicx}

\theoremstyle{plain}

\numberwithin{equation}{section}
\newtheorem{thm}{Theorem}[section]
\newtheorem{cor}[thm]{Corollary}

\newenvironment{exam}[1]%
{\begin{flushleft}\textbf{Example #1}.\enspace}%
{\end{flushleft}}

\newcounter{cond}

\newcommand{\complex}{{\mathbb C}}
\newcommand{\positive}{{\mathbb N}}
\newcommand{\real}{{\mathbb R}}
\newcommand{\ascript}{{\mathcal A}}
\newcommand{\bscript}{{\mathcal B}}
\newcommand{\cscript}{{\mathcal C}}
\newcommand{\pscript}{{\mathcal P}}
\newcommand{\qscript}{{\mathcal Q}}
\newcommand{\sscript}{{\mathcal S}}
\newcommand{\rmant}{\mathrm{ANT}}
\newcommand{\rmmix}{\mathrm{MIX}}
\newcommand{\rmmat}{\mathrm{MAT}}
\newcommand{\rmm}{\mathrm{M}}
\newcommand{\rma}{\mathrm{A}}
\newcommand{\rmcyl}{\mathrm{cyl}}
\newcommand{\muhat}{\widehat{\mu}}
\newcommand{\offspring}{\!\!\shortrightarrow}

\newcommand{\ab}[1]{\left|#1\right|}
\newcommand{\brac}[1]{\left\{#1\right\}}
\newcommand{\paren}[1]{\left(#1\right)}
\newcommand{\sqbrac}[1]{\left[#1\right]}
\newcommand{\elbows}[1]{{\left\langle#1\right\rangle}}
\newcommand{\ket}[1]{{\left|#1\right>}}
\newcommand{\bra}[1]{{\left<#1\right|}}

\begin{document}

\title{A MATTER OF MATTER\\AND ANTIMATTER
}
\author{S. Gudder\\ Department of Mathematics\\
University of Denver\\ Denver, Colorado 80208, U.S.A.\\
sgudder@du.edu
}
\date{}
\maketitle

\begin{abstract}
A discrete quantum gravity model given by a quantum sequential growth process (QSGP) is considered. The QSGP describes the growth of causal sets (causets) one element at a time in discrete steps. It is shown that the set
$\pscript$ of causets can be partitioned into three subsets $\pscript = (\rmant )\cup (\rmmix )\cup (\rmmat)$ where
$\rmant$ is the set of pure antimatter causets, $\rmmat$ the set of pure matter causets and $\rmmix$ the set of mixed matter-antimatter causets. We observe that there is an asymmetry between $\rmant$ and $\rmmat$ which may explain the matter-antimatter asymmetry of our physical universe. This classification of causets extends to the set of paths $\Omega$ in $\pscript$ to obtain $\Omega  =\Omega ^{\rmant}\cup\Omega ^{\rmmix}\cup\Omega ^{\rmmat}$. We introduce a further classification $\Omega ^{\rmmix}=\Omega _{\rmm}^{\rmmix}\cup\Omega _{\rma}^{\rmmix}$ into matter-antimatter parts. Approximate classical probabilities and quantum propensities for these various classifications are considered. Some conjectures and unsolved problems are presented.
\end{abstract}

\section{Introduction}  
This article is mainly a collection of unsolved problems concerning discrete quantum gravity described by a quantum sequential growth process (QSGP). One reason we cannot solve these problems is that we do not know the classical coupling constants and quantum dynamics of the process. Even if we knew these constants and dynamics, the problems would undoubtedly still be a challenge. Nevertheless, without this complete knowledge we can make some qualitative observations and prove some results that may be useful for this challenge. We also provide some examples that we use as test cases which indicate directions that solutions might take.

A QSGP describes the growth of causal sets (causets) one element at a time in discrete steps. One of the contributions  of this article is to observe that the set $\pscript$ of causets can be partitioned into three subsets
\begin{equation}         
\label{eq11}
\pscript =(\rmant )\cup (\rmmix )\cup (\rmmat )
\end{equation}
where we call $\rmant$ the set of pure antimatter causets, $\rmmat$ the set of pure matter causets and $\rmmix$ the set of mixed matter-antimatter causets. We next observe that there is an asymmetry between $\rmant$ and $\rmmat$. This may help to answer one of the most important questions in modern physics. Why is there a large preponderance of matter over antimatter in our physical universe? The asymmetry is most easily seen by considering the set of sequential paths $\Omega$ in $\pscript$. The elements of $\Omega$ correspond to possible universe histories. The classification of causets \eqref{eq11} can be extended to $\Omega$ to obtain the partition
\begin{equation*}
\Omega =\Omega ^{\rmant}\cup\Omega ^{\rmmix}\cup ^{\rmmat}
\end{equation*}
It turns out that the classical probabilities and quantum propensities of paths in $\Omega ^{\rmant}$ are considerably different than those in $\Omega ^{\rmmat}$. There are indications that $\Omega ^{\rmmix}$ dominates
$\Omega ^{\rmant}$ and $\Omega ^{\rmmat}$ so it is much more likely that our ``actual universe'' is in
$\Omega ^{\rmmix}$. A further classification
$\Omega ^{\rmmix}=\Omega _{\rmm}^{\rmmix}\cup\Omega _{\rma}^{\rmmix}$ partitions $\Omega ^{\rmmix}$ into matter-antimatter parts. Again, there are indications that the propensity of $\Omega _{\rmm}^{\rmmix}$ is considerably larger than the propensity of $\Omega _{\rma}^{\rmmix}$. The ratio of these propensities may be related to (or even equal to) the proportion of matter to antimatter in our universe and this may eventually be experimentally testable.

\section{Quantum Sequential Growth Processes} 
Let $x$ be a finite nonempty set. A \textit{partial order} on $x$ is a relation $<$ on $x$ that satisfies:
\begin{list} {(\arabic{cond})}{\usecounter{cond}
\setlength{\rightmargin}{\leftmargin}}
\item $a\not< a$ for all $a\in M$ (\textit{irreflexivity}).
\item If $a,b,c\in x$ with $a<b$ and $b<c$, then $a<c$ (\textit{transitivity}).
\end{list}
If $<$ is a partial order on $x$ we call $(x,<)$ a \textit{partially ordered set} or (\textit{poset}). We think of $x$ as a finite set of space-time points and $<$ as the causal order. Thus, $a<b$ if $b$ is in the causal future of $a$. For this reason we also call $x$ a \textit{causal set} or \textit{causet} \cite{blms87, sor03}. An element $a\in x$ is \textit{maximal} if there exists no $b\in x$ with $a<b$. For $a,b\in x$ we say that $a$ is an \textit{ancestor} of $b$ and $b$ is a
\textit{successor}  of $a$ if $a<b$. We say that $a$ is a \textit{parent} of $b$ and $b$ is a \textit{child} of $a$ if $a<b$ and there is no $c\in x$ such that $a<c<b$. In this work we only consider unlabeled causets and isomorphic causets are assumed to be identical.

Let $\pscript _n$ be the collection of all causets of cardinality $n$, $n=1,2,\ldots$, and let $\pscript =\cup\pscript _n$. If $x\in\pscript _n$, $y\in\pscript _{n+1}$, then $x$ \textit{produces} $y$ if $y$ is obtained from $x$ by adjoining a single element to $x$ that is maximal in $y$. We also say that $x$ is a \textit{producer} of $y$ and $y$ is an \textit{offspring} of $x$. If $x$ produces $y$ we write $x\to y$. We think of $x$ as ``growing'' into $y$ where the new element is not in the causal past of any element of $y$. We denote the set of offspring of $x$ by $x\offspring$ and for $A\in\pscript _n$ we use the notation
\begin{equation*}
A\to =\brac{y\in\pscript _{n+1}\colon x\to y, x\in A}
\end{equation*}

A \textit{path} in $\pscript$ is a string (sequence) $\omega =\omega _1\omega _2\cdots$ where
$\omega _i\in\pscript _i$ and $\omega _i\to\omega _{i+1}$, $i=1,2,\ldots\,$. An $n$-\textit{path} in $\pscript$ is a finite string $\omega _1\omega _2\cdots\omega _n$ where again $\omega _i\in\pscript _i$ and
$\omega _i\to\omega _{i+1}$. We denote the set of paths by $\Omega$ and the set of $n$-paths by $\Omega _n$. We think of $\omega\in\Omega$ as a possible universal (together with its history) and $\omega\in\Omega _n$ as a possible universe truncated at step $n$. The set of paths with initial $n$-path
$\omega =\omega _1\omega _2\cdots\omega _n\in\Omega _n$ is called an \textit{elementary cylinder set} and is denoted by $\rmcyl (\omega )$. Thus,
\begin{equation*}
\rmcyl (\omega )=\brac{\omega '\in\Omega\colon\omega '
  =\omega _1\omega _2\cdots\omega _n\omega '_{n+1}\omega '_{n+2}\cdots}
\end{equation*}
For an element $A$ of the power set $2^{\Omega _n}$ we define the \textit{cylinder set}
\begin{equation*}
\rmcyl (A)=\cup\brac{\rmcyl (\omega )\colon\omega\in A}
\end{equation*}
Thus, $\rmcyl (A)$ is the set of paths whose initial $n$-paths are elements of $A$. We use the notation
\begin{equation*}
\ascript _n=\brac{\rmcyl (A)\colon A\in 2^{\Omega _n}}
\end{equation*}
It is easy to check that the cylinder sets form an increasing sequence
$\ascript _1\subseteq\ascript _2\subseteq\cdots$ of algebras on $\Omega$ and hence
$\cscript (\Omega )=\cup\ascript _n$ is an algebra of subsets of $\Omega$. Letting $\ascript$ be the $\sigma$-algebra generated by $\cscript (\Omega )$ we have that $(\Omega ,\ascript )$ forms a measurable space.

If $x$ produces $y$ in $r$ isomorphic ways, we say that the \textit{multiplicity} of $x\to y$ is $r$ and we write
$m(x\to y)=r$. For example, in Figure~1, $m(x_3\to x_6)=2$. Let $c=(c_0,c_1,\ldots )$ be a sequence of nonnegative numbers (called \textit{coupling constants} \cite{rs00,vr06}). For $r,s\in\positive$ with $r\le s$, we define
\begin{equation*}
\lambda _c(s,r)=\sum _{k=r}^s\binom{s-r}{k-r}c_k=\sum _{k=0}^{s-r}\binom{s-r}{k}c_{r+k}
\end{equation*}
For $x\in\pscript _n$, $y\in\pscript _{n+1}$ with $x\to y$ we define the \textit{transition probability}
\begin{equation*}
p_c(x\to y)=m(x\to y)\frac{\lambda _c(\alpha ,\pi)}{\lambda _c(n,0)}
\end{equation*}
where $\alpha$ is the number of ancestors and $\pi$ the number of parents of the adjoined maximal element in $y$ that produces $y$ from $x$. It is shown in \cite{rs00, vr06} that $p_c(x\to y)$ is a probability distribution in that it satisfies the Markov-sum rule
\begin{equation*}
\sum\brac{p_c(x\to y)\colon y\in x\offspring}=1
\end{equation*}
The distribution $p_c(x\to y)$ is essentially the most general that is consistent with principles of causality and covariance \cite{rs00,vr06}. It is hoped that other theoretical principles or experimental data will determine the coupling constants. One suggestion is to take $c_k=1/k!$ \cite{sor03}. The case $c_k=c^k$ for some $c>0$ has been previously studied and is called a \textit{percolation dynamics} \cite{hen09, rs00, sur11}.

The set $\pscript$ together with the set of transition probabilities $p_c(x\to y)$ forms a
\textit{classical sequential growth process} (CSGP) which we denote by $(\pscript ,p_c)$ \cite{hen09, rs00, sor03, sur11}. It is clear that $(\pscript ,p_c)$ is a Markov chain and as usual we define the probability of an $n$-path
$\omega =\omega _1\omega _2\cdots\omega _n$ by
\begin{equation*}
p_c^n(\omega )=p(\omega _1\to\omega _2)p(\omega _2\to\omega _3)\cdots p(\omega _{n-1}\to\omega _n)
\end{equation*}
In this way $(\Omega _n,2^{\Omega _n},p_c^n)$ becomes a probability space where we define
\begin{equation*}
p_c^n(A)=\sum\brac{p_c^n(\omega )\colon\omega\in A}
\end{equation*}
for all $A\in 2^{\Omega _n}$. The probability of a causet $x\in\pscript _n$ is
\begin{equation*}
p_c^n(x)=\sum\brac{p_c^n(\omega )\colon\omega\in\Omega _n,\omega _n=x}
\end{equation*}
Of course, $x\mapsto p_c^n(x)$ is a probability measure on $\pscript _n$ and we have
\begin{equation*}
\sum _{x\in\pscript _n}p_c^n(x)=1
\end{equation*}
For $A\in\cscript (\Omega )$ of the form $A=\rmcyl (A_1)$, $A_1\in 2^{\Omega _n}$, we define $p_c(A)=p_c^n(A)$. It is easy to check that $p_c$ is a well-defined probability measure on the algebra $\cscript (\Omega )$. It follows from the Kolmogorov extension theorem that $p_c$ has a unique extension to a probability measure $\nu _c$ on the
$\sigma$-algebra $\ascript$. We conclude that $(\Omega ,\ascript ,\nu _c)$ is a probability space, the increasing sequence of subalgebras $\ascript _n$ generate $\ascript$ and that the restriction $\nu _c\mid\ascript _n=p_c^n$.

We now ``quantize'' the CSGP $(\pscript ,\nu _c)$ to obtain a quantum sequential growth process (QSGP). Let
$H=L_2(\Omega ,\ascript ,\nu _c)$ be the \textit{path Hilbert space} and $H_n=L_2(\Omega ,\ascript _n,p_c^n)$ the
$n$-\textit{path Hilbert space}, $n=1,2,\ldots\,$. Then $H_1\subseteq H_2\subseteq\cdots$ forms an increasing sequence of closed subspaces of $H$. A bounded operator $T$ on $H_n$ will also be considered as a bounded operator on $H$ by defining $Tf=0$ for every $f\in H_n^\perp$. We denote the characteristic function
$\chi _\Omega$ of $\Omega$ by $1$. Of course, $1\in H$, $\|1\|=1$ and $\elbows{1,f}=\int fd\nu _c$ for every 
$f\in H$. A $q$-\textit{probability operator} is a bounded positive operator that satisfies $\elbows{\rho 1,1}=1$. Denote the set of $q$-probability operators on $H$ and $H_n$ by $\qscript (H)$ and $\qscript (H_n)$, respectively. A sequence $\rho _n\in\qscript (H_n)$, $n=1,2,\ldots$, is \textit{consistent} if
\begin{equation*}
\elbows{\rho _{n+1}\chi _B,\chi _A}=\elbows{\rho _n\chi _B,\chi _A}
\end{equation*}
for all $A,B\in\ascript _n$. A consistent sequence $\rho _n\in\qscript (H_n)$, $n=1,2,$, is called a
\textit{discrete quantum process} (DQP). A DQP $\rho _n\in\qscript (H_n)$ on a CSGP $(\Omega ,\nu _c)$ is called a
\textit{quantum sequential growth process} (QSGP) \cite{gud111, gud112}.

A rank $1$ element of $\qscript (H)$ is called a \textit{pure} $q$-\textit{probability operator}. Thus $\rho\in\qscript (H)$ is pure if and only if $\rho$ has the form $\rho =\ket{\psi}\bra{\psi}$ for some $\psi\in H$ satisfying
$\ab{\elbows{1,\psi}}=1$ or equivalently $\ab{\int\psi d\nu _c}=1$. We then call $\psi$ a $q$-\textit{probability vector} and we denote the set of pure $q$-probability operators by $\qscript _p(H)$. A QSGP $\rho _n$ is a
\textit{pure} QSGP if $\rho _n\in\qscript _p (H_n)$, $n=1,2,\ldots\,$.

For $\rho _n\in\qscript (H_n)$ we define the $n$-\textit{decoherence functional}
$D_n\colon\ascript\times\ascript\to\complex$ by
\begin{equation*}
D_n(A,B)=\elbows{\rho _n\chi _B,\chi _A}
\end{equation*}
and this is a measure of the interference between events $A$ and $B$ when the process is described by $\rho _n$
\cite{gud111, gud112}. We define the map $\mu _n\colon\ascript\to\real ^+$ by
\begin{equation*}
\mu _n(A)=D_n(A,A)=\elbows{\rho _n\chi _A,\chi _A}
\end{equation*}
where $\mu _n(A)$ gives the quantum propensity for the occurrence of the event $A$ when the process is described by $\rho _n$. Although $\mu _n$ does not give a probability because it is not additive, it does satisfy the
\textit{grade}-2 \textit{additivity condition}: if $A,B,C\in\ascript$ are mutually disjoint, then
\begin{align}         
\label{eq21}
&\mu _n(A\cup B\cup C)\notag\\
  &=\mu _n(A\cup B)+\mu _n(A\cup C)+\mu _n(B\cup C)-\mu _n(A)-\mu _n(B)-\mu _n(C)
\end{align}
Notice that $\mu _n(\Omega )=1$ and if $\rho _n$ is a DQP then $\mu _{n+1}(A)=\mu _n(A)$ for all $A\in\ascript _n$. Since $\mu _n(A)=\|\rho _n^{1/2}\chi _A\|^2$, we conclude that $\mu _n$ is the squared norm of a vector-valued measure $A\mapsto\rho _n^{1/2}\chi _A$. In particular, if $\rho _n=\ket{\psi _n}\bra{\psi _n}$ is a pure QSGP, then
$\mu _n(A)=\ab{\elbows{\psi _n,\chi _A}}^2$ so $\mu _n$ is the squared modulus of the complex-valued measure
$A\mapsto\elbows{\psi _n,\chi _A}$.

If $\rho _n$ is a QSGP, we say that $A\in\ascript$ is \textit{suitable} if $\lim\mu _n(A)$ exists and is finite and in this case we define $\mu (A)$ to be the limit. We denote the collection of suitable sets by $\sscript (\rho _n)$. If
$A\in\ascript _n$, then $\lim\mu _n(A)=\mu _n(A)$ so $A\in\sscript (\rho _n)$ and $\mu (A)=\mu _n(A)$. In particular,
$\Omega\in\sscript (\rho _n)$ and $\mu (\Omega )=1$. This shows that the algebra
$\cscript (\Omega )\subseteq\sscript (\rho _n)$. In general, $\sscript (\rho _n)\ne\ascript$ and $\mu$ does not have a well-behaved extension from $\cscript (\Omega )$ to all of $\ascript$ \cite{djs10, gud111,hen09}. A subset $\bscript$ of
$\ascript$ is a \textit{quadratic algebra} if $\phi, \Omega\in\bscript$ and whenever $A,B,C\in\bscript$ are mutually disjoint with $A\cup B,A\cup C,B\cup C\in\bscript$, we have $A\cup B\cup C\in\bscript$. For a quadratic algebra
$\bscript$, a $q$-\textit{measure} is a map $\mu _0\colon\bscript\to\real ^+$ that satisfies the grade-2 additivity condition \eqref{eq21}. Of course, an algebra of sets is a quadratic algebra and we conclude that
$\mu _n\colon\ascript\to\real ^+$ is a $q$-measure. It is not hard to show that $\sscript (\rho _n)$ is a quadratic algebra and $\mu\colon\sscript (\rho _n)\to\real ^+$ is a $q$-measure that extends
$\mu\colon\cscript (\Omega )\to\real ^+$.

For a QSGP $\rho _n$, we call $\rho _n$ the \textit{local operators} and $\mu _n$ the \textit{local}
$q$-\textit{measures} for the process. If $\rho =\lim\rho _n$ exists in the strong operator topology, then
$\rho\in\qscript (H)$ and we call $\rho$ the \textit{global operator} for the process. If the global operator $\rho$ exists, then $\muhat (A)=\elbows{\rho\chi _A,\chi _A}$ is a (continuous) $q$-measure on $\ascript$ that extends $\mu _n$, $n=1,2,\ldots\,$. Unfortunately, the global operator does not exist in general, so we must work with the local operators \cite{djs10, gud111}. In this case, we still have the $q$-measure $\mu$ on the quadratic algebra
$\sscript (\rho _n)\subseteq\ascript$ that extends $\mu _n$, $n=1,2,\ldots\,$.

As with the coupling constants $c_n$ for a CSGP we will need additional theoretical principles or experimental data to determine the local operators $\rho _n\in\qscript (H_n)$ for a QSGP. However, we can still make some observations even with our limited knowledge. For example, suppose we are interested in the quantum propensity $\mu (A)$ of the event $A\in\sscript (\rho _n)$. Assume it is known that the classical probability $\nu _c(A)$ of $A$ occurring is small or even zero. Then the vector $\chi _A\in H$ has small norm $\|\chi _A\|=\nu _c(A)^{1/2}$. We then have
\begin{equation*}
\mu _n(A)=\elbows{\rho _n\chi _A,\chi _A}\le\|\rho _n\|\,\|\chi _A\|^2=\|\rho _n\|\nu _c(A)
\end{equation*}
If $\nu _c(A)=0$ then $\mu _n(A)=0$, $n=1,2,\ldots$, so $\mu (A)=0$. If $\|\rho _n\|$ are uniformly bounded
$\|\rho _n\|\le M$, $n=1,2,\ldots$, (which frequently happens) then $\mu (A)\le M\nu _c(A)$ which is small if $M$ is reasonable. Finally, if $A$ happens to be in $\ascript _n$ for some $n$, then
\begin{equation*}
\mu (A)=\mu _n(A)\le\|\rho _n\|\nu _c(A)
\end{equation*}
which again is small if $\|\rho _n\|$ is reasonable.

This section closes with a simple method for constructing a QSGP from a CSGP $(\pscript ,p_c)$. Although there are more general methods \cite{gud112}, the present one is instructive because it generates a quantum Markov chain. For all $x\in\pscript _n$, $y\in\pscript _{n+1}$ with $x\to y$, let $\alpha (x\to y)\in\complex$ satisfy
\begin{equation*}
\sum _{y\in x\offspring}\alpha (x\to y)p_c(x\to y)=1
\end{equation*}
We call $\alpha (x\to y)$ a \textit{transition amplitude from} $x$ \textit{to} $y$ and for\newline
$\omega =\omega _1\omega _2\cdots\omega _n\in\Omega _n$ we define the \textit{amplitude}
\begin{equation*}
\alpha (\omega )=\alpha (\omega _1\to\omega _2)\alpha (\omega _2\to\omega _3)
   \cdots\alpha (\omega _{n-1}\to\omega _n)
\end{equation*}
We next introduce the vector $\psi _n\in H$ given by
\begin{equation*}
\psi _n=\sum _{\omega\in\Omega _n}\alpha (\omega )\chi _{\rmcyl (\omega )}
\end{equation*}
It is important to notice that $\psi _n\in H_n$, $n=1,2,\ldots\,$.

\begin{thm}       
\label{thm21}
The sequence $\rho _n=\ket{\psi _n}\bra{\psi _n}$ gives a pure QSGP.
\end{thm}
\begin{proof}
To show that $\psi _n$ is a $q$-probability vector we have
\begin{align*}
\elbows{1,\psi _n}
  &=\elbows{\sum _{\omega '\in\Omega _n}\chi _{\rmcyl (\omega ')},
  \sum _{\omega\in\Omega _n}\alpha (\omega )\chi _{\rmcyl (\omega )}}
  =\sum _{\omega\in\Omega _n}\alpha (\omega )p_c^n(\omega )\\
  &=\sum _{\omega\in\Omega _n}\alpha (\omega _1\to\omega _2)p_c(\omega _1\to\omega _2)
  \cdots\alpha (\omega _{n-1}\to\omega _n)p_c^n(\omega _{n-1}\to\omega _n)\\
  &=\sum _{\omega\in\Omega _{n-1}}\!\!\alpha (\omega _1\to\omega _2)p_c(\omega _1\to\omega _2)
  \cdots\alpha (\omega _{n-2}\to\omega _{n-1})p_c^n(\omega _{n-2}\to\omega _{n-1})\\
  &\qquad\vdots\\
  &=\sum _{\omega\in\Omega _2}\alpha (\omega _1\to\omega _2)p_c^n(\omega _1\to\omega _2)=1
\end{align*}
To show that $\rho _n$ is a consistent sequence we use the notation $\omega x$ for\newline
$\omega _1\omega _2\cdots\omega _nx\in\Omega _{n+1}$ where
$\omega =\omega _1\omega _2\cdots\omega _n\in\Omega _n$ and $x\in\pscript _{n+1}$ with $\omega _n\to x$ to obtain
\begin{align*}
\sum _{x\in\omega _n\offspring}\elbows{\chi _{\rmcyl (\omega x)},\psi _{n+1}}
  &=\sum _{x\in\omega _n\offspring}
  \elbows{\chi _{\rmcyl (\omega x)},\sum _{\omega '\in\Omega _{n+1}}\alpha (\omega ')\chi _{\rmcyl (\omega ')}}\\
  &=\sum _{x\in\omega _n\offspring}\alpha (\omega x)p_c^n(\omega x)\\
  &=\sum _{x\in\omega _n\offspring}\alpha (\omega )\alpha (\omega _n\to x)p_c^n(\omega )p_c(\omega _n\to x)\\
  &=\alpha (\omega )p_c^n(\omega )=\elbows{\chi _{\rmcyl (\omega )},\psi _n}
\end{align*}
For $\omega ,\omega '\in\Omega _n$, it follows that
\begin{align*}
D_{n+1}\paren{\rmcyl (\omega ),\rmcyl (\omega ')}
  &=\elbows{\ket{\psi _{n+1}}\bra{\psi _{n+1}}\chi _{\rmcyl (\omega ')},\chi _{\rmcyl (\omega )}}\\
  &=\elbows{\chi _{\rmcyl (\omega ')},\psi _{n+1}}\elbows{\psi _{n+1},\chi _{\rmcyl (\omega )}}\\
  &=\sum _{x'\in\omega '_n\offspring}\elbows{\chi _{\rmcyl (\omega 'x')}\psi _{n+1}}
  \sum _{x\in\omega _n\to x}\elbows{\psi _{n+1},\chi _{\rmcyl (\omega x}}\\
  &=\elbows{\chi _{\rmcyl (\omega ')},\psi _n}\elbows{\psi _n,\chi _{\rmcyl (\omega )}}\\
  \noalign{\smallskip}
  &=D_n\paren{\rmcyl (\omega ),\rmcyl (\omega ')}
\end{align*}
For $A,B\in\ascript _n$ we have
\begin{equation*}
D_n(A,B)=\sum\brac{D_n\paren{\rmcyl (\omega ),\rmcyl (\omega ')}\colon\omega ,\omega '\in\ascript _n,
  \rmcyl (\omega )\subseteq A,\rmcyl (\omega ')\subseteq B}
\end{equation*}
and the result follows.
\end{proof}

The decoherence functional $D_n\colon\ascript\times\ascript\to\complex$ and the $q$-measure
$\mu _n\colon\ascript\to\real ^+$ corresponding to $\rho _n=\ket{\psi _n}\bra{\psi _n}$ are given by
\begin{align*}
D_n(A,B)&=\elbows{\psi _n,\chi _A}\elbows{\chi _B,\psi _n}\\
\mu _n(A)&=\ab{\elbows{\chi _A,\psi _n}}^2
\end{align*}
Notice that $\mu _n$ is the modulus squared of the complex measure $\lambda _n(A)=\elbows{\chi _A,\psi _n}$. Moreover, we have for all $A\in\ascript$ that
\begin{align*}
\lambda _n(A)&=\elbows{\chi _A,\sum _{\omega\in\Omega _n}\alpha (\omega )\chi _{\rmcyl (\omega )}}
  =\sum _{\omega\in\Omega _n}\alpha (\omega )\elbows{\chi _A,\chi _{\rmcyl (\omega )}}\\
  &=\sum _{\omega\in\Omega _n}\alpha (\omega )\nu _c\sqbrac{A\cap\rmcyl (\omega )}
\end{align*}
In particular, if $A\in\ascript$ then $A=\rmcyl (A_1)$ for some $A_1\in 2^{\Omega _n}$ and we have
\begin{equation*}
\lambda _n(A)=\sum _{\omega\in A_1}\alpha (\omega )p_c^n(\omega )
\end{equation*}
Moreover, for $\omega =\omega _1\omega _2\cdots\in\Omega$ we obtain
\begin{equation*}
\mu _n\paren{\brac{\omega}}
  =\ab{\alpha (\omega _1\omega _2\cdots\omega _n)p_c^n(\omega _1\omega _2\cdots\omega _n)}^2
\end{equation*}

\begin{exam}{1}
The simplest transition amplitude is $\alpha (x\to y)=1$ for every $x$ and $y$ with $x\to y$. Then
$\alpha (\omega )=1$ for every $\omega\in\Omega _n$ and
\begin{equation*}
\psi _n=\sum _{\omega\in\Omega _n}\chi _{\rmcyl (\omega )}=1
\end{equation*}
$n=1,2,\ldots\,$. The corresponding decoherence functional becomes
\begin{equation*}
D_n(A,B)=\elbows{1,\chi _a}\elbows{\chi _B,1}=\nu _c(A)\nu _c(B)
\end{equation*}
and $\mu _n(A)=\nu _c(A)^2$. Thus, $\sscript (\rho _n)=\ascript$ and $\mu (A)=\nu _c(A)^2$ is the classical probability squared. Moreover, $\ket{1}\bra{1}$ is the global operator.
\end{exam}

\begin{exam}{2}
Assume that $p(x\to y)\ne 0$ for every $x$ and $y$ with $x\to y$ and define the transition amplitude
\begin{equation*}
\alpha (x\to y)=\frac{1}{p(x\to y)\ab{x\offspring}}
\end{equation*}
where $\ab{x\offspring}$ is the cardinality of $x\offspring$. For
$\omega =\omega _1\omega _2\cdots\omega _n\in\Omega _n$, letting
\begin{equation*}
\beta _n (\omega )=\frac{1}{\ab{\omega _1\offspring}\ab{\omega _2\offspring}\cdots\ab{\omega _{n-1}\offspring}}
\end{equation*}
we have $\alpha (\omega )=1/p_c^n(\omega )\beta _n(\omega )$. If $A\in 2^{\Omega _n}$, then
\begin{equation*}
\mu\paren{\rmcyl (A)}=\mu _n\paren{\rmcyl (A)}=\ab{\sum _{\omega\in A}\beta _n(\omega )}^2
\end{equation*}
If $\omega =\omega _1\omega _2\cdots\in\Omega$ then $\mu _n\paren{\brac{\omega}}=\beta _n(\omega )^2$. Since
$\lim\beta _n(\omega )=0$, we conclude that $\brac{\omega}\in\sscript (\rho _n)$ and $\mu\paren{\brac{\omega}}=0$. It follows that if $A\subseteq\Omega$ with $\ab{A}<\infty$, then $A\in\sscript (\rho _n)$ and $\mu (A)=0$. We conjecture that $\sscript (\rho _n)\ne\ascript$ and there is no global operator.
\end{exam}

\section{Matter-Antimatter} 
Figure~1 illustrates the first four steps of a CSGP. The first four levels are complete but for lack of space, level five is not. The numbers on the arrows designate the multiplicity of the transition. Except for $x_1$ we classify the causets in terms of three types. The vertical rectangles on the left correspond to antimatter causets, the circles in the middle correspond to mixed coasts and the horizontal rectangles on the right correspond to matter causets. We whimsically call $x_1$ ``neutrino,'' $x_2$ ``positron,'' $x_3$ ``electron'' and $x_4$--$x_8$ 	``quarks'' (there are six quarks, counting multiplicity). In this way of thinking about matter-antimatter, we view a causet not only as a scaffolding for the geometry of a universe but also a placement of masses. Although there is some symmetry between antimatter and matter causets, if we take multiplicity into account there is a definite asymmetry. Since there is much more multiplicity on the matter side, this already indicates that this side might be more probable, but we shall discuss this later. We first present a rigorous definition of our classification scheme.

\includegraphics*[trim= 0 0 30 90, scale=.75, angle=90]{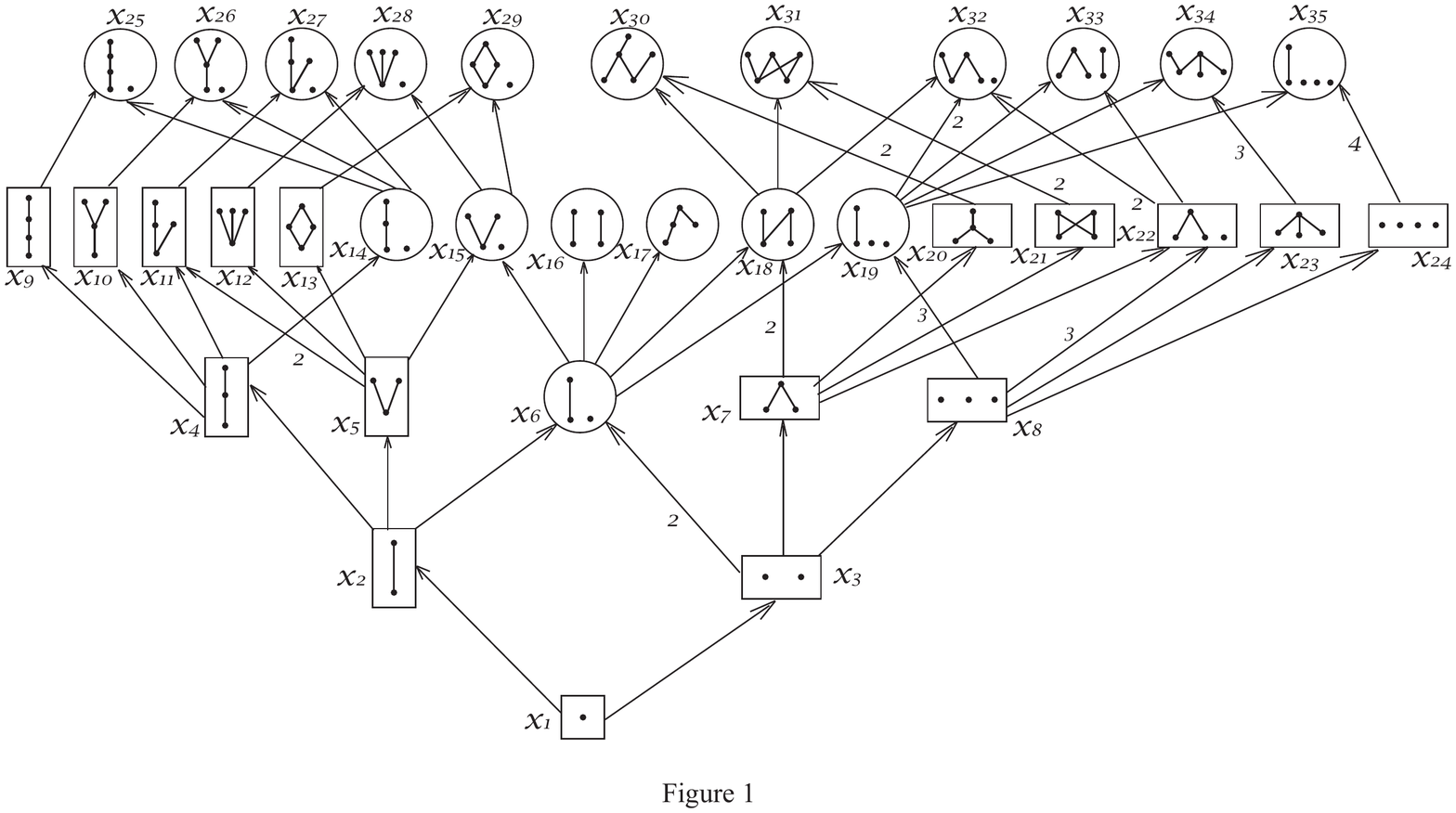}

In our classification scheme we consider $x_2$ as the source of antimatter and $x_3$ as the source of matter. We call a causet $y$ a \textit{product} of a causet $x$ if $x\ne x_1$ and either $y=x$ or there is a path containing $x$ and subsequently $y$. If $x$ is a product of $x_2$ and not $x_3$, then $x$ is an \textit{antimatter causet}. If $x$ is a product of $x_3$ and not $x_2$, then $x$ is a \textit{matter causet}. If $x$ is a product of $x_6$, then $x$ is a
\textit{mixed causet}. Notice that $x_1$ is not classified. We denote the set of antimatter, matter and mixed causets by $\rmant$, $\rmmat$ and $\rmmix$, respectively. We also use the notation $\rmant _n=\rmant\cap\pscript _n$,
$\rmmat _n=\rmmat\cap\pscript _n$, $\rmmix _n=\rmmix\cap\pscript _n$, $n=1,2,\ldots\,$. As examples we have
$x_2\in\rmant _2$, $x_3\in\rmmat _2$, $x_4,x_5\in\rmant _3$, $x_6\in\rmmix _3$, $x_7, x_8\in\rmmat _3$.

We call an element of a causet a \textit{vertex}. A vertex $a\in x$ is \textit{minimal} if there is no $b\in x$ with $b<a$. A vertex is a \textit{single-parent child} if it is the child of precisely one parent and is a \textit{multiple-parent child} if it is the child of more than one parent. No matter how complicated it is, the next result gives a method of quickly determining whether a causet is in $\rmant$, $\rmmat$ or $\rmmix$.

\begin{thm}       
\label{thm31}
{\rm (i)}\enspace A causet $x$ is a product of $x_2$ if and only if $x$ has a minimal vertex with a single-parent child.
{\rm (ii)}\enspace A causet $x$ is a product of $x_3$ if and only if $x$ has at least two minimal vertices.
{\rm (iii)}\enspace $x\in\rmmix$ if and only if $x$ has two (or more) minimal vertices at least one of which has a
single-parent child.
{\rm (iv)}\enspace $x\in\rmant$ if and only if $x$ has only one minimal vertex and $\ab{x}\ge 2$.
{\rm (v)}\enspace $x\in\rmmat$ if and only if $x$ has at least two minimal vertices and no minimal vertex has a
single-parent child
\end{thm}
\begin{proof}
(i)\enspace $x_2$ has a minimal vertex with a single-parent child and adjoining a maximal vertex to $x_2$ does not change this fact. If $x$ is a product of $x_2$ it follows by induction that $x$ has a minimal vertex with a single-parent child. Conversely, suppose $x$ has a minimal vertex $a$ with a single-parent child $b$. If $x=x_2$ then $x$ is a product of $x_2$ and we are finished. Otherwise, $x$ has a maximal vertex $c\ne a,b$. We remove $c$ and we still have $a$ and $b$ in the remaining causet. We continue this process until we obtain $x_2=\brac{a,b}$. Reversing this process gives a path containing $x_2$ and subsequently $x$. Hence, $x$ is a product of $x_2$.\newline
(ii)\enspace $x_3$ has two minimal vertices and adjoining a maximal vertex to $x_3$ does not change this fact. As in (i), we conclude that if $x$ is a product of $x_3$, then $x$ has at least two minimal vertices. Conversely, suppose $x$ has two minimal vertices $c\ne a,b$. As before, we remove $c$ and continue this process until we obtain
$x_3=\brac{a,b}$. As in (i) we conclude that $x$ is a product of $x_3$.\newline
(iii)\enspace If $x\in\rmmix$ then $x$ is a product of $x_6$ and hence $x$ is a product of both $x_2$ and $x_3$. Applying (i) and (ii) we conclude that $x$ has two minimal vertices at least one of which has a single-parent child. Conversely, suppose $x$ has two minimal vertices $a,b$ and $a$ has a single-parent child $c$. If $c$ and $b$ are the only maximal vertices of $x$, then $x=x_6$ and we are finished. Otherwise, $x$ has a maximal vertex other than $c$ or $b$ and we remove it. The resulting causet $y$ still has the minimal vertices $a,b$ and $a$ still has the
single-parent child. Continue this process until we arrive at $x_6$. Reversing the process shows that $x$ is a product of $x_6$.\newline
(iv)\enspace If $x\in\rmant$ then $\ab{x}\ge 2$ and $x\notin\rmmat$ so by (ii) $x$ has only one minimal vertex. Conversely, suppose $x$ has only one minimal vertex $a$ and $\ab{x}\ge 2$. By (ii) $x\notin\rmmat$. Now $a$ must have a child $b$ because otherwise, since $\ab{x}\ge 2$ there would be another minimal vertex. If $b$ had a parent
$c$ with $c\ne a$, then $c$ is either minimal or there is a minimal element $d$ having $c$ as a product. The first case contradicts the fact that $a$ is the only minimal vertex. In the second case, $d\ne a$ because $b$ is a child of $a$ and not of $d$. This again contradicts the fact that $a$ is the only minimal vertex. Hence, $b$ is a single-parent child and by (i), $x\in\rmant$.\newline
(v)\enspace If $x\in\rmmat$, then $x\notin\rmant$ so by (iv) $x$ must have at least two minimal vertices. Moreover, by (i) no minimal vertex has a single-parent child. Conversely, suppose $x$ has two minimal vertices and no minimal vertex has a single-parent child. By (iv), $x\notin\rmant$. By (ii), $x$ is a product of $x_3$ so $x\in\rmmat$.
\end{proof}

\begin{cor}       
\label{cor32}
A causet $x\in\rmmix$ is and only if $x$ is a product of $x_2$ and $x_3$.
\end{cor}
\begin{proof}
If $x\in\rmmix$ then clearly, $x$ is a product of $x_2$ and $x_3$. Conversely, suppose $x$ is a product of $x_2$ and $x_3$. By Theorem~\ref{thm31} (i), (ii) and (iii), $x\in\rmmix$.
\end{proof}

\begin{cor}       
\label{cor33}
{\rm (i)}\enspace $\rmant _n$, $\rmmat _n$ and $\rmmix _n$ are mutually disjoint and
\begin{equation*}
\pscript _n=\rmant _n\cup\rmmix _n\cup\rmmat _n
\end{equation*}
{\rm (ii)}\enspace $\rmant$, $\rmmat$ and $\rmmix$ are mutually disjoint and
\begin{equation*}
\pscript =\rmant\cup\rmmix\cup\rmmat
\end{equation*}
\end{cor}
\begin{proof}
This follows from Corollary~\ref{cor32}.
\end{proof}

The next corollary shows that if $x\in\rmant\cup\rmmat$ and $y\in\rmmix$ with $x\to y$, then $p(x\to y)$ has a simple form

\begin{cor}       
\label{cor34}
{\rm (i)}\enspace If $x\in\rmant$, $y\in\rmmix$ with $x\to y$, then
\begin{equation*}
p(x\to y)=\frac{m(x\to y)c_0}{\lambda _c(n,0)}
\end{equation*}
{\rm (ii)}\enspace If $x\in\rmmat$, $y\in\rmmix$ with $x\to y$, then
\begin{equation*}
p(x\to y)=\frac{m(x\to y)c_1}{\lambda _c(n,0)}
\end{equation*}
\end{cor}
\begin{proof}
Suppose $x$ produces $y$ by adjoining the maximal vertex $a$ to $x$.\newline
(i)\enspace By Theorem~\ref{thm31} (iv), $x$ has only one minimal vertex $b$ and, of course, $b$ is still minimal in
$y$. Since $y\in\rmmix$, by Theorem~\ref{thm31} (iii) $y$ must have at least two minimal elements so $a$ is maximal and minimal in $y$. Thus, $a$ is isolated in $y$ so $a$ has no parents and no ancestors. Hence,
$\lambda _c(\alpha ,\pi )=c_0$ and the result follows.\newline
(ii)\enspace By Theorem~\ref{thm31} (v) $x$ has minimal vertices $b_1,\ldots ,b_n$, $n\ge 2$, none of which has a single-parent child. Since $y\in\rmmix$ by Theorem~\ref{thm31} (iii), $a$ must be a single-parent child of precisely one of the $b_1,\ldots ,b_n$ in $y$. Thus, $x$ has one parent and one ancestor so $\lambda _c(\alpha ,\pi )=c_1$ and the result follows.
\end{proof}

There are indications that in the partition of $\pscript _n$ in Corollary~\ref{cor33} (i) that $\rmmix _n$ dominates
$\rmant _n$ and $\rmmat _n$ both numerically and probabilistically for large $n$. One reason for this is that if $x\to y$ then

\begin{list} {(\arabic{cond})}{\usecounter{cond}
\setlength{\rightmargin}{\leftmargin}}
\item $y\in\rmmix$ when $x\in\rmmix$
\item $y\in\rmant\cup\rmmix$ when $x\in\rmant$
\item $y\in\rmmat\cup\rmmix$ when $x\in\rmmat$.
\end{list}

The next result shows that the part of $\rmmix$ in (2) and (3) is nonempty. This shows that $\rmmix$ continues to
grow compared to $\rmant$ and $\rmmat$ as $n$ increases.

\begin{thm}       
\label{thm35}
{\rm (i)}\enspace If $x\in\rmant$, then $x$ produces precisely one offspring in $\rmmix$.
{\rm (ii)}\enspace If $x\in\rmmat$ has $m$ minimal vertices, then $x$ produces precisely $m$ offspring (including multiplicity) in $\rmmix$.
\end{thm}
\begin{proof}
(i)\enspace By Theorem~\ref{thm31} (iv), $x$ has precisely one minimal vertex and this vertex must have a
single-parent child. If we adjoin an isolated vertex to $x$ then by Theorem~\ref{thm31} (iii) the produced causet is mixed. Also, this is the only offspring of $x$.
(ii)\enspace By Theorem~\ref{thm31} (v), none of the minimal vertices has a single-parent child. By adjoining a
single-parent child to one of the minimal vertices, $x$ produces a mixed offspring according to
Theorem~\ref{thm31} (iii). This can be done in precisely $m$ ways.
\end{proof}

A path $\omega =\omega _1\omega _2\cdots\in\Omega$ is an \textit{antimatter} (\textit{matter}) \textit{path} if
$\omega _i\in\rmant (\rmmat )$, $i=2,3,\ldots\,$. A path $\omega =\omega _1\omega _2\cdots\in\Omega$ is
\textit{mixed} if $\omega _i\in\rmmix$ for some $i\in\brac{3,4,\ldots}$. Of course, if $\omega _i\in\rmmix$ then
$\omega _{i+1},\omega _{i+2},\ldots\in\rmmix$. We denote the sets of antimatter, matter and mixed paths by
$\Omega ^{\rmant}$, $\Omega ^{\rmmat}$, $\Omega ^{\rmmix}$, respectively. By restricting $i$ to $1,2,\ldots ,n$ we obtain the sets of $n$-paths $\Omega _n^{\rmant},\Omega _n^{\rmmat},\Omega _n^{\rmmix}\in\Omega _n$ in the natural way. It is clear that $\Omega ^{\rmant}$, $\Omega ^{\rmmat}$ and $\Omega ^{\rmmix}$ are mutually disjoint and
\begin{equation*}
\Omega =\Omega ^{\rmant}\cup\Omega ^{\rmmix}\cup\Omega ^{\rmmat}
\end{equation*}
Similarly, $\Omega _n^{\rmant}$, $\Omega _n^{\rmmat}$ and $\Omega _n^{\rmmix}$ are mutually disjoint and
\begin{equation*}
\Omega _n=\Omega _n^{\rmant}\cup\Omega _n^{\rmmix}\cup\Omega _n^{\rmmat}
\end{equation*}

Since $\rmcyl (\Omega _{n+1}^{\rmant})\subseteq\rmcyl (\Omega _n^{\rmant})$, $n=2,3,\ldots$, and
$\Omega ^{\rmant}=\cap\rmcyl (\Omega _n^{\rmant})$ we conclude that $\Omega ^{\rmant}\in\ascript$ and
\begin{equation*}
\nu _c(\Omega ^{\rmant})=\lim _{n\to\infty}p_c^n(\Omega _n^{\rmant})
\end{equation*}
Similarly, $\rmcyl (\Omega _{n+1}^{\rmmat})\subseteq\rmcyl (\Omega _n^{\rmmat})$, $n=2,3,\ldots$, and
$\Omega ^{\rmmat}=\cap\rmcyl (\Omega _n^{\rmmat})$ so that $\Omega ^{\rmmat}\in\ascript$ and
\begin{equation*}
\nu _c(\Omega ^{\rmmat})=\lim _{n\to\infty}p_c^n(\Omega _n^{\rmmat})
\end{equation*}
Moreover, $\rmcyl (\Omega _n^{\rmmix})\subseteq\rmcyl (\Omega _{n+1}^{\rmmix})$, $n=2,3,\ldots$, and
$\Omega ^{\rmmix}=\cup\rmcyl (\Omega _n^{\rmmix})$ so that $\Omega ^{\rmmix}\in\ascript$ and
\begin{equation*}
\nu _c(\Omega ^{\rmmix})=\lim _{n\to\infty}p_c^n(\Omega _n^{\rmmix})
\end{equation*}

\section{Probabilities and Propensities} 
This section computes some probabilities and notes some $q$-propensities and trends that seem to hold. One of the few probabilities that we have been able to compute precisely is $p_c^n(\Omega _n^{\rmant})$. Since
$\Omega _2^{\rmant}=\brac{x_1x_2}$ we have
\begin{equation*}
p_c^2(\Omega _2^{\rmant})=p_c(x_1\to x_2)=\frac{c_1}{c_0+c_1}
\end{equation*}
By Corollary~\ref{cor34} (i) we have
\begin{align*}
p_c^3(\Omega _3 ^{\rmant})&=p_c^3(x_1x_2x_4)+p_c^3(x_1x_2x_5)\\
  &=p_c(x_1\to x_2)\sqbrac{p_c(x_2\to x_4)+p_c(x_2\to x_5)}\\
  &=p_c(x_1\to x_2)\sqbrac{1-p_c(x_2\to x_6)}=\frac{c_1}{c_0+c_1}\sqbrac{1-\frac{c_0}{c_0+2c_1+c_2}}
\end{align*}
Since
\begin{equation*}
\Omega _4^{\rmant}=
\brac{x_1x_2x_4x_9,x_1x_2x_4x_{10},x_1x_2x_4x_{11},x_1x_2x_5x_{11},x_1x_2x_5x_{12},x_1x_2x_5x_{13}}
\end{equation*}
by Corollary~\ref{cor34}(i) we have
\begin{align*}
p_c^4&(\Omega _4^{\rmant})\\
  &=p_c(x_1\!\to\!x_2)\left\{p_x(x_2\to x_4)\sqbrac{p_c(x_4\to x_9)+p_c(x_4\to x_{10})+p_c(x_4\to x_{11})}\right.\\
  &\quad +\left.p_c(x_2\to x_5)\sqbrac{p_c(x_5\to x_{11})+p_c(x_5\to x_{12})+p_c(x_5\to x_{13})}\right\}\\
  &=p_c(x_1\to x_2)\left\{p_c(x_2\to x_4)\sqbrac{1-p_c(x_4\to x_{14})}\right.\\
  &\quad \left.+p_c(x_2\to x_5)\sqbrac{1-p_c(x_5\to x_{15})}\right\}\\
  &=p_c(x_1\to x_2)\sqbrac{1-p_c(x_2\to x_6)}\sqbrac{1-p_c(x_4\to x_{14})}\\
  &=\frac{c_1}{c_0+c_1}\paren{1-\frac{c_0}{c_0+2c_1+c_2}}\paren{1-\frac{c_0}{c_0+3c_1+3c_2+c_3}}
\end{align*}
Continuing by induction we obtain
\begin{equation}         
\label{eq41}
p_c^n(\Omega _n^{\rmant})=\frac{c_0}{c_0+c_1}\paren{1-\frac{c_0}{c_0+2c_1+c_2}}
  \cdots\paren{1-\frac{c_0}{\displaystyle\sum _{i=0}^{n-1}\binom{n-1}{i}c_i}}
\end{equation}
The sequence \eqref{eq41} converges to $\nu _c(\Omega ^{\rmant})$.

To get an idea of the probabilities $\nu _c(\Omega ^{\rmant})$ consider a percolation dynamics with
$c_k=c^k$, $c>0$. Then \eqref{eq41} becomes
\begin{equation}         
\label{eq42}
p_c^n(\Omega _n^{\rmant})=\frac{c_1}{1+c}\sqbrac{1-\frac{1}{(1+c)^2}}\sqbrac{1-\frac{1}{(1+c)^3}}
  \cdots\sqbrac{1-\frac{1}{(1+c)^{n-1}}}
\end{equation}
Table~1 gives $\nu _c(\Omega ^{\rmant})$ for various values of $c$. Notice that, as a function of $c$,
$\nu _c(\Omega ^{\rmant})$ increases monotonically until it reaches a maximum of $0.3035$ at about
$c=1.3$ and then decreases monotonically.
\vglue 2pc

\begin{center}
\begin{tabular}{c|c|c|c|c|c|c|c|c}
$c$&$0.2$&$0.3$&$0.5$&$0.7$&$0.9$&$1.0$&$1.1$&$1.2$\\
\hline
$\nu _c(\Omega ^{\rmant})$&$0.00357$&$0.03122$&$0.1385$&$0.2264$
  &$0.2753$&$0.2888$&$0.2972$&$0.3018$\\
\end{tabular}
\end{center}
\vskip 1pc
\begin{center}
\begin{tabular}{c|c|c|c|c|c|c|c}
$c$&$1.3$&$1.4$&$1.5$&$1.6$&$2.0$&$3.0$&$4.0$\\
\hline
$\nu _c(\Omega ^{\rmant})$&$0.3035$&$0.3031$&$0.3012$&$0.2982$
  &$0.2801$&$0.2295$&$0.1901$\\
\noalign{\bigskip}
\multicolumn{8}{c}%
{\textbf{Table 1}}\\
\end{tabular}
\end{center}
\vglue 2pc

We suspect that for this model, the $q$-propensities will be smaller. For instance in Example~1, the maximum values of the $q$-propensity is
\begin{equation*}
\mu(\Omega ^{\rmant})=(0.3035)^2=0.09211
\end{equation*}

We next consider the factorial dynamics $c_k=1/k!$. In this case \eqref{eq41} becomes
\bigskip
\begin{equation*}
p_c^n(\Omega _n^{\rmant})=\frac{1}{2}\prod _{j=2}^{n-1}
  \sqbrac{1-\frac{1}{\displaystyle\sum _{i=0}^j\binom{j}{i}\frac{1}{i!}}}
\end{equation*}
\bigskip

\noindent Table~2 gives $p_c^n(\Omega _n^{\rmant})$ for various values of $n$.
\vglue 1pc

\bigskip\parindent=0pt
\begin{center}
\begin{tabular}{c|c|c|c|c}
$n$&$10$&$20$&$50$&$100$\\
\hline
$p_c^n(\Omega _n^{\rmant})$&$0.20793$&$0.19522$&$0.19363$&$0.19362$\\
\noalign{\medskip}
\multicolumn{5}{c}%
{\textbf{Table 2}}\\
\end{tabular}
\end{center}
\bigskip\parindent=18pt

\noindent According to Mathematica
\begin{equation*}
\nu_c(\Omega ^{\rmant})=\lim _{n\to\infty}p_c^n(\Omega _n^{\rmant})=0.19361
\end{equation*}
We have found $\nu _c(\Omega ^{\rmant})$ for other coupling constants. For $c_k=1/(k!)^k$,
$c_k=1/((k+1)^k$, $c_k=1/(k!)^{50}$ we have that $\nu _c(\Omega ^{\rmant})$ equals $0.093313$, $0.032498$,
$0.0049751$, respectively.

Due to multiplicities, the situation for $p_c^n(\Omega _n^{\rmmat})$ becomes much more complicated and we have not found this quantity precisely. However, we can point out some tendencies. Proceeding as we did for
$\Omega _n^{\rmant})$ we obtain
\begin{align*}
p_c^2(\Omega _2^{\rmmat})&=\frac{c_0}{c_0+c_1}\\
  p_c^3(\Omega _3^{\rmmat})&=\frac{c_0}{c_0+c_1}
  \sqbrac{1-\frac{2c_1}{c_0+2c_1+c_2}}\\\noalign{\medskip}
  p_c^4(\Omega _4^{\rmmat})&=\frac{c_0}{c_0+c_1}
  \sqbrac{1-\frac{c_1}{c_0+2c_1+c_2}\paren{1+\frac{2c_2+3c_0}{c_0+3c_1+3c_2+c_3}}}
\end{align*}
and the terms become increasingly more complicated.

We now compare the first few terms of $p_c^n(\Omega _n^{\rmmat})$ with $p_c^n(\Omega _n^{\rmant})$. For the percolation dynamics with $c=1$ we have
\begin{align*}
p_c^2(\Omega _2^{\rmmat})&=p_c^2(\Omega _2^{\rmant})=\tfrac{1}{2}\\
  p_c^3(\Omega _3^{\rmmat})&=\tfrac{1}{4},\quad p_c^3(\Omega _3^{\rmant})=\tfrac{5}{14}\\
  p_c^4(\Omega _4^{\rmmat})&=\tfrac{19}{64},\quad p_c^4(\Omega _4^{\rmant})=\tfrac{21}{64}\\
\end{align*}
For the factorial dynamics $c_k=1/k!$ we have
\begin{align*}
p_c^2(\Omega _2^{\rmmat})&=p_c^2(\Omega _2^{\rmant})=\tfrac{1}{2}\\
  p_c^3(\Omega _3^{\rmmat})&=\tfrac{3}{14},\quad p_c^3(\Omega _3^{\rmant})=\tfrac{3}{8}\\
  p_c^4(\Omega _4^{\rmmat})&=\tfrac{61}{238},\quad p_c^4(\Omega _4^{\rmant})=\tfrac{70}{238}\\
\end{align*}
The previous data indicates that the tendency is for $p_c^n(\Omega _n^{\rmmat})<p_c^n(\Omega _n^{\rmant})$. This is also consistent with Theorem~\ref{thm35} with shows that $\ab{\Omega _n^{\rmmat}}$ is much smaller than
$\ab{\Omega _n^{\rmant}}$ for large $n$. When this tendency holds, both $\nu _c(\Omega ^{\rmmat})$ and
$\nu _c(\Omega ^{\rmant})$ are small and this indicates that $\mu (\Omega ^{\rmmat})$ and $\mu (\Omega ^{\rmant})$ are also likely to be small.

\begin{exam}{3}
This example gives further evidence for the tendency $p_c^n(\Omega _n^{\rmmat})<p_c^n(\Omega _n^{\rmant})$. Let $\omega _a^n\in\Omega _n$ be the extreme antimatter path $\omega _a^n=x_1x_2x_4x_9\cdots y$. Then
\begin{equation*}
p_c^n\paren{\omega _a^n}=\paren{\frac{c_1}{c_0+c_1}}\paren{\frac{c_1+c_2}{c_0+2c_1+c_2}}\cdots
  \paren{\frac{\displaystyle\sum _{j=1}^n\binom{n-1}{j-1}c_j}{\displaystyle\sum _{j=0}^n\binom{n}{j}c_j}}
\end{equation*}
If $\omega _m^n\in\Omega _n$ is the extreme matter path $\omega _m^n=x_1x_3x_8x_{24}\cdots z$ we have
\begin{equation*}
p_c^n\paren{\omega _m^n}=\paren{\frac{c_0}{c_0+c_1}}\paren{\frac{c_0}{c_0+2c_1+c_2}}\cdots
  \paren{\frac{c_0}{\displaystyle\sum _{j=0}^n\binom{n}{j}c_j}}
\end{equation*}
For the percolation dynamics $c_k=c^k$ we have
\begin{align*}
p_c^n(\omega _a^n)&=\paren{\frac{c}{1+c}}^n\\
  p_c^n(\omega _m^n)&=\frac{c^n}{(1+c)^{n(n+1)/2}}=\paren{\frac{c}{1+c}}^n\frac{1}{(1+c)^{n(n-1)/2}}
\end{align*}
Denoting the ``complete'' paths by $\omega _a,\omega _m\in\Omega$ we conclude that
\begin{align*}
\nu_c(\omega _a)&=\lim _{n\to\infty}p_c^n(\omega _a^n)=0\\
\nu_c(\omega _m)&=\lim _{n\to\infty}p_c^n(\omega _m^n)=0\\
\end{align*}
It follows that $\brac{\omega _a},\brac{\omega _m}\in\sscript (\rho _n)$ and $\mu (\omega _a)=\mu (\omega _m)=0$ for any $q$-propensity in the percolation dynamics. Moreover, not only is $p_c^n(\omega _m^n)$ much smaller than
$p_c^n(\omega _a^n)$ for large $n$ but we have
\begin{equation*}
\lim _{n\to\infty}\frac{p_c^n(\omega _m^n)}{p_c^n(\omega _a^n)}=0
\end{equation*}
These same general results hold for the factorial dynamical $c_k=1/(k!)$.
\end{exam}

As we have seen, there are indications that $\Omega ^{\rmmix}$ dominates $\Omega ^{\rmant}$ and
$\Omega ^{\rmmat}$. We now examine $\Omega ^{\rmmix}$ more closely. We can partition $\Omega ^{\rmmix}$ into two parts
\begin{align*}
\Omega _{\rma}^{\rmmix}&=\brac{\omega =\omega _1\omega _2\cdots\in\Omega ^{\rmmix}\colon\omega _2=x_2}\\
\Omega _{\rmm}^{\rmmix}&=\brac{\omega =\omega _1\omega _2\cdots\in\Omega ^{\rmmix}\colon\omega _2=x_3}
\end{align*}
and as usual we define $\Omega _{\rma ,n}^{\rmmix}$, $\Omega _{\rmm ,n}^{\rmmix}$ as the truncated $n$-paths from
$\Omega _{\rma}^{\rmmix}$, $\Omega _{\rmm}^{\rmmix}$, respectively. Since
$\Omega _{\rma ,3}^{\rmmix}=\brac{x_1x_2x_6}$ and $\Omega _{\rmm ,3}^{\rmmix}=\brac{x_1x_3x_6}$, we have that
\begin{align*}
p_c^3(\Omega _{\rma ,3}^{\rmmix})&=\frac{c_0c_1}{(c_0+c_1)(c_0+2c_1+c_2)}\\
p_c^3(\Omega _{\rmm ,3}^{\rmmix})&=\frac{2c_0c_1}{(c_0+c_1)(c_0+2c_1+c_2)}
\end{align*}
Hence, for any value of the coupling constants we have
\begin{equation*}
\frac{p_c^3(\Omega _{\rmm ,3}^{\rmmix})}{p_c^3(\Omega _{\rma ,3}^{\rmmix})}=2
\end{equation*}
Continuing we have that
\begin{align*}
\Omega _{\rma ,4}^{\rmmix}&=\left\{x_1x_2x_4x_{14},x_1x_2x_5x_{15}, x_1x_2x_6x_{14},x_1x_2x_6x_{15},\right.\\
 &\qquad \left.x_1x_2x_6x_{16},x_1x_2x_6x_{17},x_1x_2x_6x_{18},x_1x_2x_6x_{19}\right\}\\
\Omega _{\rmm ,4}^{\rmmix}&=\left\{x_1x_3x_6x_{14},x_1x_3x_6x_{15}, x_1x_3x_6x_{16},x_1x_3x_6x_{17},\right.\\
 &\qquad \left.x_1x_3x_6x_{18},x_1x_3x_6x_{19},x_1x_3x_7x_{18},x_1x_3x_8x_{19}\right\}
\end{align*}
Hence,
\begin{align*}
p_c^4(\Omega _{\rma ,4}^{\rmmix})&=p_c^3(x_1x_2x_6)+p_c^4(x_1x_2x_4x_{14})+p_c^4(x_1x_2x_5x_{15})\\
  &=\frac{c_0c_1}{\lambda _c(1,0)\lambda _c(2,0)}
  +\frac{c_0c_1(2c_1+c_2)}{\lambda _c(1,0)\lambda _c(2,0)\lambda _c(3,0)}\\\noalign{\medskip}
p_c^4(\Omega _{\rmm ,4}^{\rmmix})&=p_c^3(x_1x_3x_6)+p_c^4(x_1x_3x_7x_{18})+p_c^4(x_1x_3x_8x_{19})\\
  &=\frac{2c_0c_1}{\lambda _c(1,0)\lambda _c(2,0)}
  +\frac{c_0c_1(3c_1+2c_2)}{\lambda _c(1,0)\lambda _c(2,0)\lambda _c(3,0)}
\end{align*}
and the ratio becomes
\begin{equation*}
\frac{p_c^4(\Omega _{\rmm ,4}^{\rmmix})}{p_c^4(\Omega _{\rma ,4}^{\rmmix})}
  =\frac{5c_0+6c_1+8c_2+2c_3}{c_0+5c_1+4c_2+c_3}
\end{equation*}
In the percolation dynamics the ratio $r_c^4$ becomes
\begin{equation*}
r_c^4=\frac{5+6c+8c^2+2c^3}{1+5c+4c^2+c^3}
\end{equation*}
In particular, $r_1^4=21/11$ and for $c\ge 1$, $r_c^4\approx 2$. Moreover, for small $c$ we have $r_c^4\approx 5$. For the factorial dynamics $c_k=1/k!$ the ratio becomes $r^4\approx 1.88$.

We shall not continue this line of thought to find $p_c^5(\Omega _{\rma ,5}^{\rmmix})$ and
$p_c^5(\Omega _{\rmm ,5}^{\rmmix})$ because it becomes much more involved. However, we can observe that
\begin{align*}
\nu _c(\Omega _{\rma}^{\rmmix})&=\lim _{n\to\infty}p_c^n(\Omega _{\rma ,n}^{\rmmix})\\
\nu _c(\Omega _{\rmm}^{\rmmix})&=\lim _{n\to\infty}p_c^n(\Omega _{\rmm ,n}^{\rmmix})
\end{align*}
For any coupling constants, it is clear that
\begin{equation*}
p_c^n(\Omega _{\rmm ,n}^{\rmmix})>p_c^n(\Omega _{\rma ,n}^{\rmmix})
\end{equation*}
so that the ratio
\begin{equation*}
r=\frac{\nu _c(\Omega _{\rmm}^{\rmmix})}{\nu _c(\Omega _{\rma}^{\rmmix})}>1
\end{equation*}
We conjecture that $r$ is relatively large. Moreover, once we know the DQP $\rho _n$, $n=1,2,\ldots$, we conjecture that $\Omega _{\rma}^{\rmmix},\Omega _{\rmm}^{\rmmix}\in\sscript (\rho _n)$ and that $r$ is related to
$\mu (\Omega _{\rmm}^{\rmmix})/\mu(\Omega _{\rma}^{\rmmix})$. Finally this latter ratio may be the ratio of matter to antimatter in our physical universe.

\end{document}